\documentclass[a4paper]{article}
\usepackage{charter}
\usepackage{fullpage}
\usepackage[american]{babel}
\usepackage{amsmath}
\usepackage{amssymb}
\usepackage{amsthm}
\usepackage{ulem}
\usepackage{graphicx}
\newtheorem{theorem}{Theorem}
\newtheorem{corollary}[theorem]{Corollary}
\newtheorem{definition}[theorem]{Definition}
\newtheorem{lemma}[theorem]{Lemma}
\newcommand{\equals}{\stackrel{\mathrm{def}}{=}}
\newcommand{\lcm}{\mathrm{lcm}}

\title{Gang FTP scheduling of periodic \\and parallel rigid real-time tasks}
\author{Jo\a"el Goossens\hfill Vandy Berten\\
Universit\a'e libre de Bruxelles (U.L.B.)\\
CP212, 50 av. F.D. Roosevelt\\
1050 Brussels, Belgium\\
\{joel.goossens,vandy.berten\}@ulb.ac.be}

\begin{document}
\maketitle
\thispagestyle{empty}

\begin{abstract}
In this paper we consider the scheduling of periodic and parallel rigid tasks. We provide (and prove correct) an \emph{exact} schedulability test for Fixed Task Priority (FTP) Gang scheduler sub-classes: Parallelism Monotonic, Idling, Limited Gang, and Limited Slack Reclaiming. Additionally, we study the predictability of our schedulers: we show that Gang FJP schedulers are not predictable and we identify several sub-classes which are actually predictable. Moreover, we extend the definition of rigid, moldable and malleable jobs to recurrent tasks.  
\end{abstract}

\section{Introduction}\label{sec:intro}
We consider the preemptive scheduling of real-time tasks on identical multiprocessor platforms (see~\cite{baker2,Baker2005An-analysis-of-}). We deal with \emph{parallel} real-time tasks, the case where each job may be executed on different processors \emph{simultaneously}, i.e., we allow \emph{job parallelism}. Nowadays, the design of parallel programs is common thanks to parallel programming paradigms like Message Passing Interface (MPI~\cite{Gorlatch1998A-Generic-MPI-I,Lusk1999Using-MPI-:-por}) or Parallel Virtual Machine (PVM~\cite{Sunderam1990PVM:-A-Framewor,Geist1994PVM:-Parallel-V}). Even better, sequential programs can be parallelized using tools like OpenMP (see~\cite{355074} for details).
\paragraph{Related Work.}
Few models and results in the literature concern hard real-time and parallel tasks. \mbox{Manimaran} et al.\@ in~\cite{mani98} consider the \emph{non-preemptive} EDF scheduling of periodic tasks. \mbox{Han} et al.\@ in~\cite{han89} considered the scheduling of a
(finite) set of real-time jobs allowing job parallelism while we consider the scheduling of either infinite set of jobs or equivalently a set a periodic tasks. In previous work we contributed mainly to the \emph{feasibility} problem of parallel tasks. In~\cite{CCG06b} we provided a task model which integrates job parallelism. We proved that the time-complexity of the feasibility problem of these systems is linear relatively to the number of (sporadic) tasks. More recently, we considered the scheduling of jobs which are composed of phases to be executed sequentially ; in~\cite{Berten2009Feasibility-Tes} we provided a \emph{necessary} feasibility test. Regarding the \emph{schedulability} of recurrent real-time tasks, and to the best of our knowledge, we can only report the S.~Kato et al.\@ work (see~\cite{kato2009gang} for details) which considers the Gang scheduling of rigid tasks (the number of processors is fixed beforehand) and provides a \emph{sufficient} schedulability condition for Gang EDF scheduling.
\paragraph{This Research.}
In this paper we study the scheduling of recurrent and parallel rigid tasks. Our main contribution is an \emph{exact} schedulability test for Fixed Task Priority (FTP) Gang schedulers. Additionally, we study the predictability of our schedulers: we show that Gang FJP schedulers are not predictable and we identify several sub-classes which are actually predictable. Our technique is based on previous work for the scheduling of periodic tasks upon multiprocessors~\cite{dateLCJG07,etfaLCJG06}. To summarize the technique, we characterized for FTP schedulers and for asynchronous periodic task models, \emph{upper bounds} of the first time-instant where the schedule repeats. Based on the upper bounds and the \emph{predictability} property, we provide \emph{exact} schedulability tests for asynchronous constrained deadline periodic task sets. The predictability property is important in the technique and will be revisited in this work.
We also extend the definition of rigid, moldable and malleable jobs to \emph{recurrent} tasks.   

\paragraph{Paper Organization.}
This paper is organized as follows. Section~\ref{sec:model} introduces definitions, the model of computation and our assumptions. In Section~\ref{sec:predict} we study the predictability of our system, in particular we show that Gang FJP schedulers are not predictable and we identify several sub-classes which are actually predictable. We prove the periodicity of feasible schedules of periodic systems in Section~\ref{sec:periodicity}. In Section~\ref{sec:exact} we combine the periodicity and predictability properties, to provide, for our Gang FTP sub-classes, an \emph{exact} schedulability test. Lastly, we conclude in Section~\ref{sec:conclusion}.

\section{Model and Definitions}\label{sec:model}
\subsection{Parallel Terminology}
The parallel literature~\cite{Drozdowski2005Scheduling-Para, Buyya99, Feitelson96towardconvergence} defines several kind of parallel \emph{tasks}. But \emph{tasks} in the non real-time parallel terminology does not have the same meaning as \emph{tasks} in real-time scheduling literature. Actually, \emph{tasks} in the parallel literature corresponds to \emph{jobs} in our real-time community (i.e., corresponds to task \emph{instance}). Especially, the notion of rigid recurrent task is not defined and does not extend trivially from the non real-time literature, in this section we fill the gap.

\begin{definition}[Rigid, Moldable and Malleable Job]
A \emph{job} is said to be:
\begin{description}
\item[Rigid] if the number of processors assigned to this job is specified externally to the scheduler a priori, and does not change throughout its execution;
\item[Moldable] if the number of processors assigned to this job is determined by the scheduler, and does not change throughout its execution;
\item[Malleable] if the number of processors assigned to this job can be changed by the scheduler during the job's execution.
\end{description}
\end{definition}

\begin{definition}[Rigid, Moldable and Malleable Recurrent Task]\label{def:taskparallel}
A periodic/sporadic \emph{task} is said to be:
\begin{description}
\item [Rigid] if all its jobs are rigid, and the number of processors assigned to the jobs is specified externally to the scheduler;
\item [Moldable] if all its jobs are moldable;
\item [Malleable] if all its jobs are malleable.
\end{description}
\end{definition}

Notice that a rigid task does not necessarily have jobs with the same size. For instance, if the user/application decides that odd instances require $v$ processors, and even instances $v'$ processors, the task is said to be rigid.

\subsection{Task and Job Model}
We consider the preemptive scheduling of parallel jobs on a multiprocessor platform with $m$ processors. We will focus on the problem of scheduling a set of parallel jobs, each job $J_j \equals (r_{j}, v_{j}, e_{j}, d_{j})$ is characterized by a release time $r_j$, $v_{j}$ a required number of processors, an execution requirement $e_j$ and an absolute  deadline $d_j$. The job $J_j$ must execute for $e_j$ time units over the interval $[r_j,d_j)$ on $v_{j}$ processors. We consider the scheduling of \emph{rigid} tasks since $v_{i}$ is fixed externally to the scheduler. Actually, S.~Kato et al.\@ in~\cite{kato2009gang} have the same assumption as we do, and, given the Definition~\ref{def:taskparallel}, they consider \emph{rigid} tasks, and not \emph{moldable} tasks, as said in their paper --- otherwise the scheduler would determine $v_{i}$ on-line and at job-level. 

As we will consider periodic systems, let $\tau = \{\tau_{1}, \ldots, \tau_{n}\}$ denote a set of $n$ periodic parallel tasks. Each task $\tau_{i} = (O_{i}, v_{i}, C_{i}, D_{i}, T_{i},)$ will generate an infinite number of jobs, where the $k^\text{th}$ job of task $\tau_i$ is 
$$(O_i+ (k-1)T_i, v_i, C_i, O_i+ (k-1)T_i+D_i).$$

The execution requirement of a job of $\tau_{i}$ corresponds as a $C_{i} \times v_{i}$ \emph{rectangle}.
In this document we assume $D_{i} \leq T_{i}$ for any $\tau_{i}$, i.e., we consider constrained deadline systems. We consider multiprocessor platforms $\pi$ composed of $m$ identical processors: $\{\pi_1, \pi_2, \ldots, \pi_m \}$.

\subsection{Priority Assignment and Schedulers}

In this document we consider FTP and FJP schedulers with the following definitions.

\begin{definition}[FTP]
  A priority assignment  is a \emph{Fixed Task Priority}  assignment if it assigns the priorities to the tasks beforehand; at  run-time each job priority corresponds to its task priority. (An FTP scheduler uses an FTP priority assignment.)
\end{definition}

We assume that tasks are indexed according to priority (lower the index, higher the priority). 

\begin{definition}[FJP]
A priority assignment is a \emph{Fixed Job Priority}  assignment if and only if it satisfies the condition that: for every pair of jobs $J_i$ and $J_j$, if $J_i$ has higher priority than $J_j$ at some time-instant, then $J_i$ always has higher priority than $J_j$. (An FJP scheduler uses an FJP priority assignment.) 
\end{definition}

Remark that any FTP assignment is also FJP.


\begin{definition}[Gang FJP] At each instant, the algorithm schedules jobs on processors as follows: the highest priority (active) job $J_i$ is scheduled on the first $v_{i}$ available processors (if any). The very same rule is then applied to the remaining active jobs on the remaining available processors.
\end{definition}
\paragraph{Priority inversion.} Figure~\ref{fig:gangftp} illustrates the Gang FTP schedule ($\tau_{1} $ is the highest priority task and $\tau_{3}$ the lowest one) of $\tau_{1} = (0,2,2,5,5), \tau_{2} = (0,2,3,5,5), \tau_{3} = (0,1,4,5,5)$. Notice that Gang FJP and FTP can produce schedules where a lower-priority job ($J_{j}$) is scheduled while an active higher-priority job ($J_{i}$) is not (typically if $v_{i}>v_{j}$ --- which occurs at time $0$ in our example, $\tau_{2}$ and $\tau_{3}$ are active, $\tau_{3}$ is executing in $[0,2)$ while $\tau_{2}$ is not). This phenomenon, called \emph{priority inversion} in this document, could be a drawback as we will see. Fortunately, we keep an important FTP-property: the scheduling of the sub-set $\{\tau_{1}, \ldots, \tau_{i}\}$ is not jeopardized by lower-priority tasks ($\{\tau_{i+1}, \ldots, \tau_{n}\}$).

\begin{figure}
\begin{center}
\includegraphics[width=.8\linewidth]{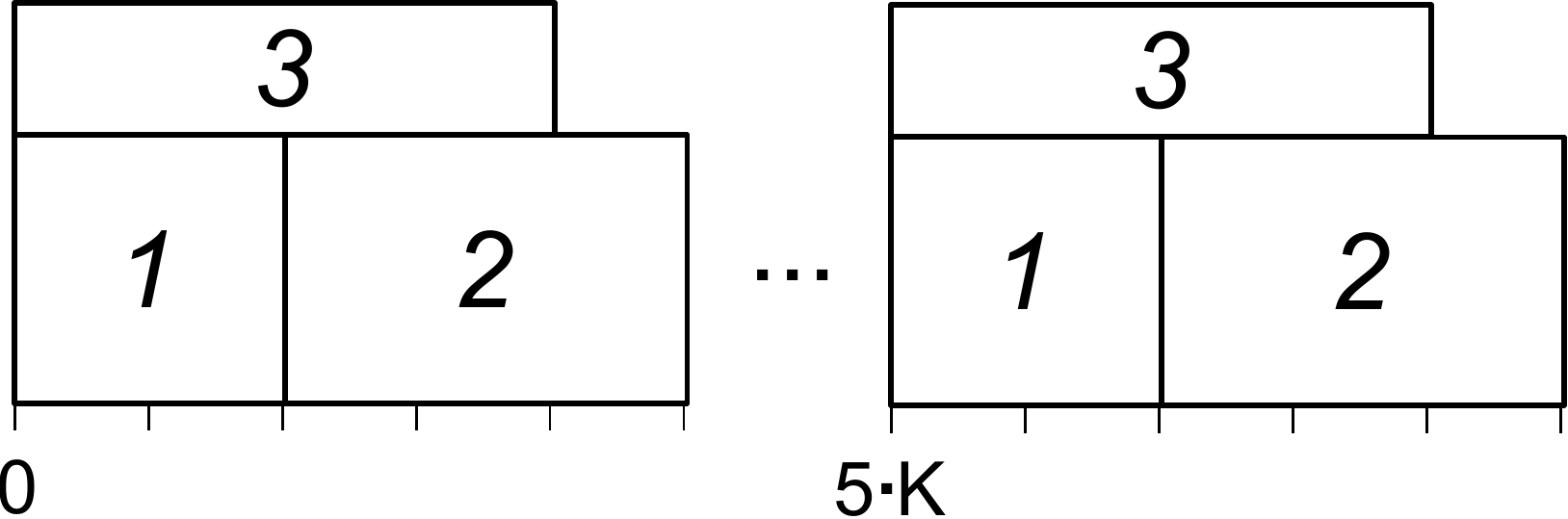}
\caption{\label{fig:gangftp}Gang FTP schedule with priority inversion at time $0$.} 
\end{center}\end{figure}

\begin{definition}[Schedule $\sigma(t)$]\label{def:sched}
  For any set of jobs $J \equals \{J_{1}, J_{2}, J_{3},\ldots\}$ and any set of $m$ identical processors $\{\pi_1, \ldots, \pi_m \}$ we define the \emph{schedule} $\sigma(t)$ of system $\tau $ at time-instant $t$ as
  $\sigma : \mathbb{N} \rightarrow \mathbb{N}^m$ where $\sigma(t) \equals (\sigma_1(t),
  \sigma_2(t), \ldots,
  \sigma_m(t))$ with \\
  $\sigma_j(t) \equals \left\{
\begin{array}{ll}
0, & \text{if there is no job scheduled on } \pi_j \\
& \text{at time-instant } t; \\
i, & \text{if job } J_i \mbox{ is scheduled on } \pi_j \text{ at}\\
& \text{time-instant } t.  
\end{array}
\right.$
\end{definition}
\begin{definition}[Availability of the processors]\label{def:avaiJob}
For any ordered set of jobs $J$ and any set of $m$ processors $\{\pi_1, \ldots, \pi_m \}$, we define the {\em availability of the processors} $A(J,t)$ of the set of jobs $J$ at time-instant $t$ as the set of available processors: $A(J,t) \equals \{j \mid \mbox{ } \sigma_j(t)=0 \}$, where $\sigma$ is the schedule of $J$.
\end{definition}

\begin{definition}[Active, Ready and Running jobs]
A job is said to be \emph{active} if it has been released, but is not finished yet. An active job is \emph{ready} if it is not currently served ; an active job is \emph{running} otherwise. 
\end{definition}

\section{Predictability of Gang Scheduling}\label{sec:predict}

We consider the scheduling of sets of job $J \equals J_{1}, J_{2}, J_{3}\ldots$, (finite or infinite set of jobs) and without loss of generality we consider jobs in a decreasing order of priorities $(J_1 > J_2 > J_{3} > \cdots$). We suppose that the execution time of each job $J_i$ can be any value in the interval $[e_i^{-}, e_i^{+}]$ and we denote by $J^{+}_i$ the job defined as $J^{+}_i \equals (r_i,v_{i},e_i^{+},d_i)$. We denote by $J^{(i)}$ the set of the first $i$ higher priority jobs. We denote also by $J^{(i)}_{-}$ the set $\{ J^{-}_1, \ldots, J^{-}_i \}$ and by $J^{(i)}_{+}$ the set $\{J^{+}_1, \ldots, J^{+}_i \}$. Let $S(J)$ be the time-instant at which the lowest priority job of $J$ begins its execution in the schedule. Similarly, let $F(J)$ be the time-instant at which the lowest priority job of $J$ completes its execution in the schedule.

\begin{definition}[Predictable algorithms]
A scheduling algorithm is said to be {\em predictable} if $S(J^{(i)}_{-}) \leq S(J^{(i)}) \leq S(J^{(i)}_{+})$ and $F(J^{(i)}_{-}) \leq F(J^{(i)}) \leq F(J^{(i)}_{+})$, for all $i \geq 1$ and for all schedulable $J^{(i)}_{+}$ sets of jobs.
\end{definition}

Notice that the predictability of an algorithm implies that any system schedulable when all tasks use their worst case execution time is also schedulable when a task takes less time than expected. We then just need to show that the system is schedulable in the worst case to prove that the system is schedulable in all scenarios.

In previous work~\cite{Cucu-Grosjean2009Predictability-} we proved, for a quite general model, i.e., FJP priority schedulers on unrelated multiprocessors, the predictability for \emph{sequential} jobs (i.e., $v_{i}=1$ for any $J_{i}$). Unfortunately that property is not satisfied for parallel jobs.

\begin{lemma}\label{lem:notPredict}
Gang FJP schedulers are not predictable on multiprocessors.
\end{lemma}

\begin{figure}
\begin{center}
\includegraphics[width=.8\linewidth]{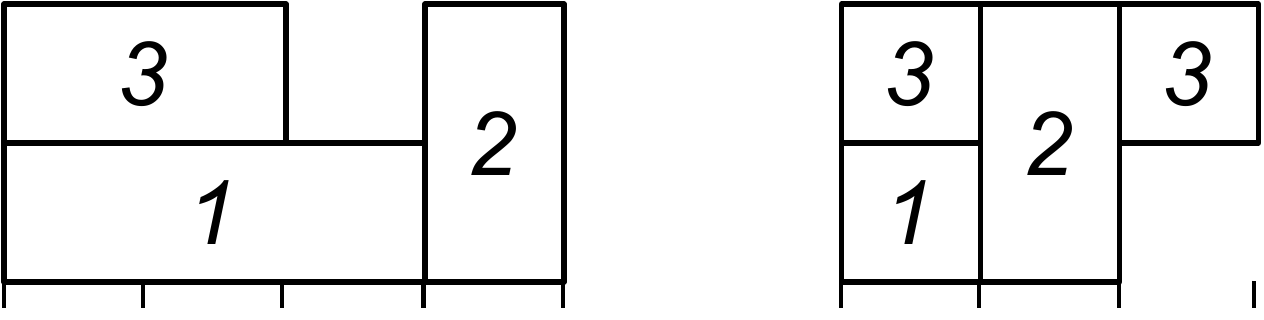}
\caption{\label{fig:notPredictable}Non-predictability of Gang FJP schedulers. $1 > 2 > 3$, and they all arrive at time $0$.} 
\end{center}\end{figure}

\begin{proof}
Here is an example task system, on 2 processors (see Figure~\ref{fig:notPredictable}):
\[ J_{1} = (0, 1, 3, 3), J_{2} = (0, 2, 1, 4), J_{3} = (0, 1, 2, 2)\enspace.\]
Upon two processors and using the priority assignment $J_{1}>J_{2}>J_{3}$, Gang FJP schedules the set of jobs ($J_{3}$ completes at time-instant 2). Unfortunately, if the actual duration of $J_{1}$ is 1, $J_{2}$ will preempt $J_{3}$ at time $t=1$ and $J_{3}$ will complete \emph{later}, at time-instant 3. Then, $J_3$ does not miss its deadline in the ``worst case'' scenario, but misses it if $J_1$ uses less than its worst case execution time $e_1$.
\end{proof}

This negative result implies that neither the DM, RM nor EDF are predictable for Gang scheduling. 

The problem we highlight in this example occurs because some jobs which were not preempted in the worst case scenario are preempted in a scenario with shorter execution times. In other words, by allowing a job $J_i$ taking advantage of some slack time given by another (higher priority) job, $J_i$ interrupts a job $J_j$ (with $J_j<J_i$) which would not have been suspended if we did not have any slack.

If, as we will do and prove in this work, we find a way to avoid the priority inversion phenomenon, we will never have those problematic preemptions. Indeed, if no lower priority job $J_j$ is authorized to start between the arrival of $J_i$ ($J_j<J_i$) and its start time, then if $J_i$ starts anywhere between its arrival time, and its start time in the worst case scenario, it will not interrupt tasks that it would not have interrupt in the worst case.

The ``problematic preemptions'' can be avoided by (at least) three ways:
\begin{itemize}
\item By avoiding the priority inversion;
\item By avoiding any slack (or by not using it);
\item By using the slack, but in a ``smart'' way.
\end{itemize}

In order to obtain a predictable system, we identify two ways of modifying the system:
\begin{itemize}
\item First, we will propose to constraint the priority assignment. We introduce the \emph{Parallelism Monotonic FTP assignment}, and prove the predictability of these priority assignments;
\item Second, we will propose three variants of the scheduler, giving a predictable behavior. Those variants are the \emph{idling scheduler} (not using the slack), the \emph{limited Gang FJP scheduler} (avoiding priority inversion), and the \emph{Gang FJP scheduler with limited slack reclaiming} (smartly using the slack).
\end{itemize}

\subsection{Parallelism Monotonic FTP Assignment}
In this section we will consider a sub-class of Gang FTP assignments which are predictable.
\begin{definition}[Parallelism Monotonic]
An FTP priority assignment is \emph{Parallelism Monotonic} iff $i<j \Rightarrow v_{i} \leq v_{j}$. 
\end{definition}

Notice that this class is very interesting from the theoretical point of view, but might be not a good choice for some implementation. Indeed, it gives a low priority to highly parallel jobs, which makes them more difficult to schedule. In general, it might be useful to first schedule the very parallel jobs, and then to fill the available processors with smaller jobs.

We will now prove that any Parallelism Monotonic assignment are predictable.

In~\cite{Cucu-Grosjean2009Predictability-} we showed that $A(J^{(i)}_{+},t) \subseteq A(J^{(i)},t)$, for all $t$ and all $i$. In other words, that at any time-instant the processors available in ${\sigma^{(i)}_{+}}$ are also available in ${\sigma^{(i)}}$. The counterexample used in the proof of Lemma~\ref{lem:notPredict} violates that property as well. In the following we will consider another kind of processors availability.

\begin{definition}[Level-($i$) availability of the processors]
For any ordered set of $i-1$ jobs $J=J_{1}, \ldots,J_{i-1}$ and any set of $m$ processors, we define the \emph{level-($i$) availability of the processors} $A_{i}(J,t)$ of the set of jobs $J$ at time-instant $t$ as follows 
\[ A_{i}(J,t) \equals \begin{cases} 
	\#A(J,t) & \text{if $\#A(J,t)\geq v_{i}$};\\
 	0 & \text{otherwise}.
\end{cases}
\]
\end{definition}

Informally speaking, $A_{i}(J,t)$ is the the number of available processors if the latter is sufficient to schedule $J_{i}$, otherwise $A_{i}(J,t)$ is null.  

\begin{lemma}\label{lem:AParMon} 
For any schedulable ordered set of jobs $J$, using a Gang FJP Parallelism Monotonic on $m$  processors, we have $A_{i}(J^{(i-1)}_{+},t) \leq A_{i}(J^{(i-1)},t)$, for all $t$ and all $i$. (We consider that the sets of jobs are ordered in the same decreasing order of the priorities, i.e., $J_1 > J_2 > \cdots > J_{\ell}$ and $J_1^{+} > J_2^{+} > \cdots > J_{\ell}^{+}$.)
\end{lemma}
\begin{proof}
  The proof is made by induction on $\ell$ (the number of jobs).  Our inductive hypothesis is the following: $A_{k}(J^{(k-1)}_{+},t) \leq A_{k}(J^{(k-1)},t)$, for all $t$ and $1 < k \leq i+1$.
 
  The property is true in the base case since $A_{2}(J^{(1)}_{+},t) \leq A_{2}(J^{(1)},t)$, for all $t$. Indeed, $S(J^{(1)}) = S(J^{(1)}_{+})$. Moreover $J_{1}$ and $J_{1}^{+}$ are both scheduled on the (same) first $v_{1}$ processors, but $J_{1}^{+}$ will be executed for the same or a greater amount of time than $J_{1}$.

  We will show now that $A_{i+2}(J^{(i+1)}_{+},t) \leq A_{i+2}(J^{(i+1)},t)$, for all $t$.

  Since the jobs in $J^{(i)}$ have higher priority than $J_{i+1}$, then the scheduling of $J_{i+1}$ will not interfere with higher priority jobs which have already been scheduled. Similarly, $J^{+}_{i+1}$ will not interfere with higher priority jobs of $J^{(i)}_{+}$ which have already been scheduled. Therefore, we may build the schedule $\sigma^{(i+1)}$ from $\sigma^{(i)}$, such that the jobs $J_1, J_2, \ldots, J_{i}$, are scheduled at the very same instants and on the very same processors as they were in $\sigma^{(i)}$. Similarly, we may build $\sigma^{(i+1)}_{+}$ from $\sigma^{(i)}_{+}$.

Note that property is straightforward for time-instants where $J^{(i+1)}$ is not scheduled since the processor availability is not modified and by definition of level-($i+2$) processor availability. 

We will consider time-instant $t$, from $r_{i+1}$ to the completion of $J_{i+1}$ (which is actually not after the completion of $J_{i+1}^{+}$, see below for a proof), we distinguish between three cases:
\begin{enumerate}

\item $A_{i+1}(J^{(i)}_{+},t)  = A_{i+1}(J^{(i)},t) = 0$: in both situations no enough processors are available for $J_{(i)}$ (and $J_{(i)}^{+}$). Therefore, both jobs, $J_{i+1}$ and $J_{i+1}^{+}$, do not progress and we obtain $A_{i+2}(J^{(i+1)}_{+},t) = A_{i+2}(J^{(i)}_{+},t) = A_{i+2}(J^{(i+1)},t) = A_{i+2}(J^{(i)},t) = 0$, since $v_{i+2} \geq v_{i+1}$. The progression of $J_{i+1}$ is identical to $J_{i+1}^{+}$. 
\item \label{item:add} $0 = A_{i+1}(J^{(i)}_{+},t) < v_{i+1} \leq A_{i+1}(J^{(i)},t)$: $J_{i+1}$ progress on the $v_{i+1}$ first available processors in $A(J^{(i)},t)$ (not available in $A(J^{(i)}_{+},t)$). $J_{i+1}^{+}$ does not progress. $A_{i+2}(J^{(i)}_{+},t) = 0$ since $v_{i+2} \geq v_{i+1}$. The progression of $J_{i+1}$ is strictly larger than $J_{i+1}^{+}$.
\item \label{item:idem} $v_{i+1} \leq A_{i+1}(J^{(i)},t) \leq A_{i+1}(J^{(i)}_{+},t)$, $J_{i+1}$ and $J_{i+1}^{+}$ progress on the same processors. The property follows by induction hypothesis. 
\end{enumerate}
Therefore, we showed that $A_{i+2}(J^{(i+1)}_{+},t) \leq A_{i+2}(J^{(i+1)},t)$, for all $t$, from $r_{i+1}$ to the completion of $J_{i+1}$ and that $J_{i+1}$ does not complete after $J_{i+1}^{+}$. For the time-instant after the completion of $J_{i+1}$ the property is trivially true by induction hypothesis.
\end{proof}

\begin{theorem}\label{thm:predParMon} Gang FJP schedulers are predictable on identical platforms with Parallelism Monotonic priority assignment.
\end{theorem}
\begin{proof}
In the framework of the proof of Lemma~\ref{lem:AParMon} we actually showed extra properties which imply that Gang FJP Parallelism Monotonic schedulers are predictable on identical platforms: (i) $J_{i+1}$ completes not after $J_{i+1}^{+}$ and (ii) $J_{i+1}$ can be scheduled either at the very same instants as $J_{i+1}^{+}$ or may progress during \emph{additional} time-instants (case~(\ref{item:add}) of the proof) these instants may precede the time-instant where $J_{i+1}^{+}$ commences its execution.
\end{proof}

\subsection{Idling Scheduler}

Instead of giving constraints on the priority assignment, we can also adapt our scheduler in order to make it predictable. A first way of doing that is to force tasks to run \emph{exactly} up to their worst case. If a task does not use its worst case, then the scheduler idles the processor(s) up to the expected end time.

\begin{definition}[Idling scheduler]
An \emph{idling scheduler} idles any processor that was used by a task which finished earlier than its worst case, up to the time the processor would have been released in the worst case scenario.
\end{definition}

\begin{lemma}\label{lemma:predIdling} Gang FJP Idling schedulers are predictable on identical platforms.
\end{lemma}

\begin{proof}
The proof is very straightforward in this case: any task starts at the same time in the worst case $J^+$, and in the case where some tasks use less than than their worst case. And if we consider the completion time of a job as the time at which its (possibly empty) idle period finishes, then the end time will be the same in the worst case scenario as in any case. Then,

$$S(J^{(i)}_{-}) \leq S(J^{(i)}) \leq S(J^{(i)}_{+}),$$ 
and 
$$F(J^{(i)}_{-}) \leq F(J^{(i)}) \leq F(J^{(i)}_{+}).$$
\end{proof}

\subsection{Limited Gang FJP Scheduler}

The predictability can also be ensured if we avoid the priority inversion phenomenon reported in Section~\ref{sec:intro}, more precisely by restricting Gang FTP/FJP as follows:

\begin{definition}[Limited Gang FJP scheduler] At each instant, the algorithm schedules jobs on processors as follows: the highest priority (active) job $J_i$ is scheduled on the first $v_{i}$ available processors (if any). The very same rule is then applied to the remaining active jobs on the remaining available processors \emph{only if} $J_{i}$ was scheduled (i.e., if at least $v_{i}$ processors were available).
\end{definition}

We now prove that limited Gang FJP are predictable but first an additional definition.

\begin{definition}[Limited level-($i$) availability of the processors]
For any ordered set of $i-1$ jobs $J=J_{1}, \ldots,J_{i-1}$ and any set of $m$ processors, we define the \emph{limited level-($i$) availability of the processors} $\hat{A}_{i}(J,t)$ of the set of jobs $J$ at time-instant $t$ as follows ($\hat{A}_{0}(J,t)= m$ for all $J,t$):
\[ \hat{A}_{i}(J,t) \equals \begin{cases} 
	0 			& \text{if $\hat{A}_{i-1}(J,t) = 0$};\\
	\#A(J,t) 	& \text{if }\#A(J,t)\geq v_{i} \text{ and } \\
		& \hat{A}_{i-1}(J,t) \neq 0;\\
 	0 & \text{otherwise.}
\end{cases}
\]
\end{definition}

\begin{lemma}\label{lem:ALimGang} 
For any schedulable ordered set of jobs $J$, using a Limited Gang FJP on $m$  processors, we have $\hat{A}_{i}(J^{(i-1)}_{+},t) \leq \hat{A}_{i}(J^{(i-1)},t)$, for all $t$ and all $i$. (We consider that the sets of jobs are ordered in the same decreasing order of the priorities, i.e., $J_1 > J_2 > \cdots$ and $J_1^{+} > J_2^{+} > \cdots$.)
\end{lemma}
\begin{proof}
The property follows using a similar reasoning as the proof of Lemma~\ref{lem:AParMon} and the fact that $\hat{A}_{i}(J^{(i-1)},t) \geq \hat{A}_{i+1}(J^{(i-1)},t).$
\end{proof}

\begin{theorem}\label{thm:predLimited}
Limited Gang FJP schedulers are predictable on identical platforms.
\end{theorem}

\begin{proof}
The proof is similar to the proof given for Theorem~\ref{thm:predParMon}. $\hat{A}_{i}(J^{(i-1)},t)$ describes the number of processors available to schedule $J_i$. As this number is at any time higher in $J_+^{(i-1)}$ than in $J^{(i-1)}$, then $J_i$ will never start later in  $J^{(i-1)}$ than in $J_+^{(i-1)}$, and will never finish later either.
\end{proof}

Remark that by using limited Gang scheduling we accept to lose efficiency in the resource utilization to ensure system predictability. 
 
\subsection{Gang FJP and Limited Slack Reclaiming}
As highlighted previously, the problem of using the slack caused by a job finishing earlier than expected is that it could cause a preemption that would not have occurred if the job had used its worst case execution time. But with a closer look, we can see that the problem only occurs when a job $J_i$ wider than a job $J_j$ ($v_i>v_j$) takes advantage of the slack created by $J_j$ early completion. A way of avoiding this is to only allow job not larger than the early completed job to use the slack. This is what we propose in this technique.

\begin{definition}[Slack server]
A \emph{slack server} of level $\ell$, width $w$ and length $\lambda$ is a job of priority $\ell$, on $w$ processors, running for $\lambda$ units of time, serving jobs with a priority lower than $\ell$ which do not require more than $w$ processors. If no task are available, the server stays idle until the end of the $\lambda$ units of time.
\end{definition}

It may be noticed that:
\begin{itemize}
\item We do not give any constraint about the way the ``slack server scheduling'' (the way jobs are scheduled inside the slack server) is performed;
\item Within a slack server, we might run several jobs in parallel, as long as they never need more than $w$ processors simultaneously;
\item If a job being served by the slack server becomes eligible by the ``global scheduler'', then it should be interrupted in the slack server, and made available to the global scheduler;
\item All jobs served by the slack server should still stay in the ready state (but not running) from the ``global scheduler'' point of view.
\end{itemize}

Regarding this definition of a slack server, we can now define how our schedule will work.

\begin{definition}[Gang FJP and limited slack reclaiming]
A \emph{Gang FJP scheduler with limited slack reclaiming}, works as follows: At each scheduling point (the completion of a job or an arrival):
\begin{itemize}
\item If this corresponds to the end of a job $J_i$, and this job used $e'<e_i$ units of time, starts a slack server of level $i$, width $v_i$, length $e_i-e'$;
\item Otherwise, the highest priority (active) job $J_i$ is scheduled on the first $v_{i}$ available processors (if any). The very same rule is then applied to the remaining active jobs on the remaining available processors.
\end{itemize}
\end{definition}

We can make an important observation about this scheduling algorithm. Jobs that are run in the slack server will not have any other impact on the global scheduler that reducing the execution time of those jobs. So if we consider the slack server as a black box, the schedule will be exactly the same as if this black box was just idling. The only impact will be the proportion between actual and worst case execution time.

Remark also that will we cause priority inversion: some jobs will be run inside the slack server, while other higher priority ready job (but wider than the slack server) will be kept suspended.

\begin{figure}
\begin{center}
\includegraphics[width=.9\linewidth]{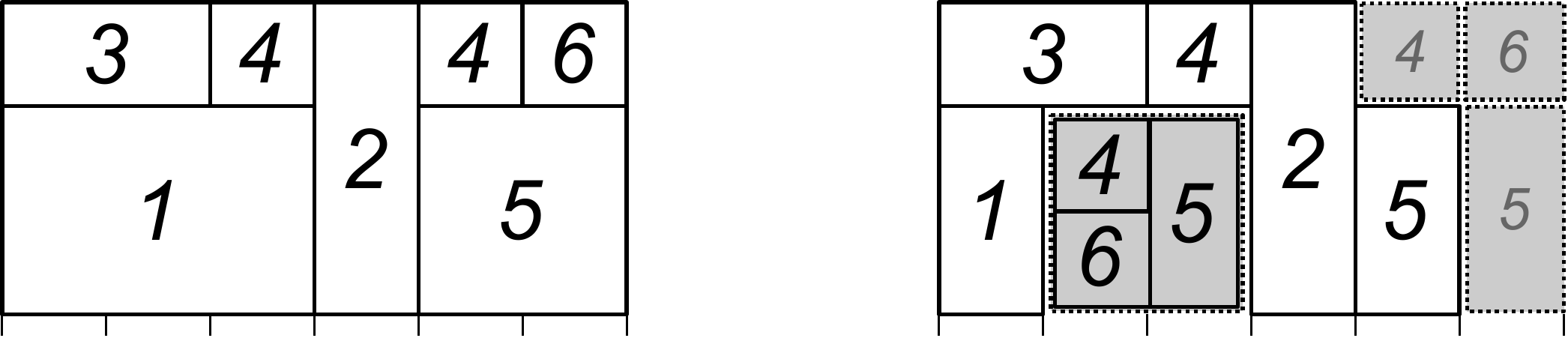}
\caption{\label{fig:slackServer}Slack server example. Left: all jobs use their worst case execution time. Right: Job 1 is shorter than expected, and a slack server is set up (gray part).} 
\end{center}\end{figure}

Figure~\ref{fig:slackServer} illustrates how the slack server works. We consider the following set of jobs (they all have the same arrival time $0$, and the same deadline $6$:
\begin{center}
\begin{tabular}{|c|c|c|c|c|c|c|}
\hline
			& $J_1$		& $J_2$ 	&	$J_3$	& $J_4$		& $J_5$		& $J_6$ \\ \hline
$v_i$ 	& 2        	&  3			& 1 			& 1			& 2			& 1 \\ \hline
$e_i$ 	& 3        	&  1			& 2 			& 2			& 2			& 1 \\ \hline
\end{tabular} 
\end{center}

The left side of Figure~\ref{fig:slackServer} shows the schedule where all jobs use their worst case execution time.
On the right side, $J_1$ finishes at time $1$ (instead of $3$). The schedulers launches then a slack server (in gray on the figure) of level $1$, width $2$ and length $2$, in order to fill the space that would have been used by $J_1$ in the worst case scenario. This server is then scheduled as a job of priority $1$, as was $J_1$. At time $1$, the slack server sees that jobs $J_2$, $J_4$, $J_5$, $J_6$ are \emph{ready} ($J_3$ is running). But $J_2$ is too wide, so the slack server can for instance choose (arbitrarily) to run $J_4$ and $J_6$. After one unit of time, the scheduler sees that it can run $J_4$, the highest priority task which can run on the only available processor released by the end of $J_3$. $J_4$ is then ``preempted'' inside the slack server, and run normally. Then the slack server choses to run $J_5$ ($J_6$ is done). At time $3$, the slack server ends, and $J_5$ is preempted. 

At time $4$ and $5$, we need to start slack servers for $J_4$, $J_5$ and $J_6$, but they do not receive any work to perform.

Remark that if we compare both schedules of Figure~\ref{fig:slackServer}, and see the slack server as part of the concerned job, all tasks start end and at the same time in both scenarios.

\begin{theorem}\label{thm:slack} Gang FJP schedulers with limited slack reclaiming are predictable on identical platforms.
\end{theorem}

\begin{proof}
From the schedulability point of view, this behaves exactly the same way as the Idling server. But instead of being idle, the slack server decreases the actual execution time of some ready (but not running) jobs.

One job will never preempt a job that would not have been preempted in the worst case scenario.
\end{proof}

Notice that with this kind of scheduler, we might considered the system as being not fully FTP anymore. Some jobs are indeed eligible (enough resource to run them), but are left waiting. In the presentation of this section, we said that we did not give any constraint on the scheduler of the slack server. Indeed, the method we use does not have any impact on the predictability of the system. But we can of course use a FTP scheduling algorithm. This does not make the global system to be strictly FTP, but it makes closer.

Notice also that we present a system with two level of scheduler: one global, and one inside the slack server. This distinction was used for the sake of presentation, but in a real implementation, the global scheduler can of course also do the job of the slack server scheduler.

\section{Periodicity}\label{sec:periodicity}
In this section we prove the periodicity of feasible Gang FTP schedules. It is important to note that we assume in this section that each job of the same task (say $\tau_{i}$) has an execution requirement which is \emph{exactly} $C_{i}$ time units. Thanks to the predictability property this situation corresponds to the worst case.

Remark that, as we consider only the case where all job has its worst case, the schedule of an idling, slack reclaiming or general scheduler is exactly the same.

\begin{theorem}\label{thm:asynPer} For any preemptive (limited or not) Gang FTP algorithm $\mathcal{A}$, if an asynchronous constrained deadline system $\tau$ is $\mathcal{A}$-feasible, then the $\mathcal{A}$-feasible schedule of $\tau$ on $m$ identical processors is periodic with a period $P \equals \lcm\{T_{1}, \ldots, T_{n}\}$ from instant $S_n$ where $S_i$ is defined inductively as follows:

  \begin{itemize}
  \item $S_1 \equals O_1$; 
  \item $S_i \equals \max \{ O_i, O_i+ \left\lceil \dfrac{S_{i-1}-O_i}{T_i}
    \right\rceil T_i \}, \forall i \in \{2,3, \ldots, n \}$.
  \end{itemize}
  
  (Assuming that the execution time of each task is constant.)
\end{theorem}

\begin{proof}
  The proof is made by induction on $n$ (the number of tasks). We denote by $\sigma^{(i)}$ the schedule obtained by considering only the task subset $\tau^{(i)}$, the first higher priority $i$ tasks $\{\tau_1, \ldots, \tau_i \}$, and by $A(J^{(i)}, t)$ the corresponding availability of the processors. Our inductive hypothesis is the following: the schedule $\sigma^{(k)}$ is periodic from $S_k$ with a period $P_k \equals \lcm\{T_{1}, \ldots, T_{k}\}$ for all $1 \leq k \leq i$.

  The property is true in the base case: $\sigma^{(1)}$ is periodic from $S_1=O_1$ with period $P_1$, for $\tau^{(1)}= \{\tau_1 \}$:  since we consider (feasible) constrained deadline systems, at instant $P_1=T_1$ the  previous request of $\tau_1$ has finished its execution and the  schedule repeats.

  We shall now show that any $\mathcal{A}$-feasible schedule of $\tau^{(i+1)}$ is  periodic with period $P_{i+1}$ from $S_{i+1}$.

  Since $\sigma^{(i)}$ is periodic with a period $P_{i}$ from $S_{i}$ the  following equation is verified:

\begin{equation}
  \label{eq:stateInter}
\sigma^{(i)}(t)=\sigma^{(i)}(t+P_i), \forall t \geq S_{i}.
\end{equation}

We denote by $S_{i+1} \equals \max \{ O_{i+1}, O_{i+1}+ \left\lceil \dfrac{S_{i}-O_{i+1}}{T_{i+1}} \right\rceil T_{i+1} \}$ the first request of $\tau_{i+1}$ not before $S_i$.

Since the tasks in $\tau^{(i)}$ have higher priority than $\tau_{i+1}$, then the scheduling of $\tau_{i+1}$ will not interfere with higher priority tasks which are already scheduled. Therefore, we may build $\sigma^{(i+1)}$ from $\sigma^{(i)}$ such that the tasks $\tau_1, \tau_2, \ldots, \tau_i$ are scheduled at the very same instants and on the very same processors as they were in $\sigma^{(i)}$. We apply now the induction step: for all $t \geq S_{i}$ in $\sigma^{(i)}$ we have $A(J^{(i)},t) = A(J^{(i)}, t +P_i)$ the availability of the processors repeats. Notice that at those instants $t$ and $t+P_i$ the available processors (if any) are the same. Consequently at only these instants where $\#A(J^{(i)},t)\ge v_{i+1}$, task $\tau_{i+1}$ {\em may} be executed. Notice that the scheduler can decide to leave one or several processor(s) to be idle intentionally in a deterministic and memoryless way. Notice also, in the ``non limited case'', that $\tau_{i+1}$ might start executing before a higher priority task $\tau_j$ (with $j<i+1$), if $v_j>\#A(J^{(i)},t)>v_{i+1}$. But as soon as $v_j$ processors are available in $A(J^{(i)}, t)$, $\tau_{i+1}$ is preempted (if still running) and the CPU is given to $\tau_j$.  

The instants $t$ with $S_{i+1} \leq t < S_{i+1}+P_{i+1}$, where $\tau_{i+1}$ may be executed in $\sigma^{(i+1)}$, are periodic with period $P_{i+1}$ since $P_{i+1}$ is a multiple of $P_i$.  Moreover since the system is feasible and we consider constrained deadlines, the only active request of $\tau_{i+1}$ at $S_{i+1}$, respectively at $S_{i+1}+P_{i+1}$, is the one activated at $S_{i+1}$, respectively at $S_{i+1}+P_{i+1}$. Consequently, the instants at which the deterministic and memoryless algorithm $\mathcal{A}$ schedules $\tau_{i+1}$ are periodic with period $P_{i+1}$. Therefore, the schedule $\sigma^{(i+1)}$ repeats from $S_{i+1}$ with period equal to $P_{i+1}$ and the property is true for all $1 \leq k \leq n$, in particular for $k=n:$ $\sigma^{(n)}$ is periodic with period equal to $P$ from $S_n$ and the property follows.
\end{proof}

\section{Exact Schedulability Test}\label{sec:exact}
Now we have the material to define an exact schedulability test for rigid and asynchronous  periodic systems.

\begin{corollary}\label{cor:exact-test} For any preemptive Gang FTP predictable algorithm $\mathcal{A}$ (i.e., Parallelism Monotonic, Idling, Limited Gang, and Limited Slack Reclaiming variants) and for any asynchronous rigid constrained deadline system $\tau$ on $m$ identical processors, $\tau$ is $\mathcal{A}$-schedulable if and only if 
\begin{itemize}
\item all deadlines are met in $[0, S_n+P)$ and
\item $\theta(S_{n})=\theta(S_{n}+P)$
\end{itemize}
where $S_i$ are defined inductively in Theorem~\ref{thm:asynPer}.
\end{corollary}
\begin{proof}
Corollary~\ref{cor:exact-test} is a direct consequence of Theorem~\ref{thm:asynPer} and the predictability of Parallelism Monotonic (Theorem~\ref{thm:predParMon}), Idling (Lemma~\ref{lemma:predIdling}), Limited Gang (Theorem~\ref{thm:predLimited}), and Limited Slack Reclaiming (Theorem~\ref{thm:slack}) variants.
\end{proof}

\section{Conclusion and Future Work}\label{sec:conclusion}
In this paper we considered the scheduling of periodic and parallel rigid tasks. We provided and proved correct an \emph{exact} schedulability test for Fixed Task Priority (FTP) Gang scheduler sub-classes: Parallelism Monotonic, Idling, Limited Gang, and Limited Slack Reclaiming. Additionally, we studied the predictability of our schedulers: we show that Gang FJP schedulers are not predictable and we identify several sub-classes which are actually predictable. We also extended the definition of rigid, moldable and malleable jobs to recurrent tasks. 

In future work we aim to extend the model by considering \emph{moldable} tasks --- task can be executed in a varying number of processors --- that is the scheduler can determine, on-line, the rectangle of each task instance (job) based upon parallel performance model (e.g., the one defined in \cite{CCG06b}).

\bibliographystyle{acm}
\bibliography{biblio.bib}
\end{document}